\documentclass[a4paper,UKenglish,cleveref,autoref,numberwithinsect]{lipics-v2019}
\title{Around finite second-order coherence spaces} %TODO Please add

\author{Lê Thành D\~{u}ng Nguy\~{ê}n}{LIPN,
  UMR 7030 CNRS, Université Paris 13, Sorbonne Paris Cité, France \and \url{https://nguyentito.eu/}}{nltd@nguyentito.eu}{https://orcid.org/0000-0002-6900-5577}{}%mandatory, please use full name; only 1 author per
\authorrunning{L.~T.~D.~Nguy\~{ê}n}%mandatory. First: Use abbreviated first/middle names. Second (only in severe cases): Use first author plus 'et al.'

\Copyright{John Q. Public and Joan R. Public}%TODO mandatory, please use full first names. LIPIcs license is "CC-BY";  http://creativecommons.org/licenses/by/3.0/



\category{}%optional, e.g. invited paper

\relatedversion{\url{https://arxiv.org/abs/1902.00196}}
\supplement{}%optional, e.g. related research data, source code, ... hosted on a repository like zenodo, figshare, GitHub, ...

\funding{Partially supported by the French ANR project Elica
  (ANR-14-CE25-0005).}%optional, to capture a funding statement, which applies to all authors. Please enter author specific funding statements as fifth argument of the \author macro.

\acknowledgements{Thomas Seiller made several suggestions to significantly
  improve the categorical presentation and conjectured
  \cref{thm:polynomialgrowth}. The proof in \cref{sec:regular} incorporates
  major simplifications found by Kazushige Terui. Thanks also to Seng Beng Goh,
  Thomas Ehrhard, Damiano Mazza, Paul-André Melliès, Paolo Pistone, Pierre
  Pradic and Lorenzo Tortora de Falco for instructive discussions; and to the
  anonymous reviewers of a previous submission for their useful
  feedback.}%optional

\nolinenumbers %uncomment to disable line numbering

\hideLIPIcs  %uncomment to remove references to LIPIcs series (logo, DOI, ...), e.g. when preparing a pre-final version to be uploaded to arXiv or another public repository

\EventEditors{John Q. Open and Joan R. Access}
\EventNoEds{2}
\EventLongTitle{42nd Conference on Very Important Topics (CVIT 2016)}
\EventShortTitle{CVIT 2016}
\EventAcronym{CVIT}
\EventYear{2016}
\EventDate{December 24--27, 2016}
\EventLocation{Little Whinging, United Kingdom}
\EventLogo{}
\SeriesVolume{42}
\ArticleNo{23}
\usepackage{cmll}
\usepackage{stmaryrd}
\usepackage{xspace}
\usepackage{todonotes}
\usepackage{bussproofs}  \EnableBpAbbreviations

\usepackage{tikz}
\usetikzlibrary{shapes,arrows}
\usepackage{tikz-cd}

\newcommand{\enquote}[1]{``#1''}

\newcommand{\llfragment}[2][]{\ensuremath{\textsc{#2}_{#1}}\xspace}

\newcommand{\affimplies}{\rightarrowtriangle}

\newcommand{\Lang}{\mathcal{L}}
\newcommand{\STlam}{{\mathrm{ST}\lambda}}

\newcommand{\malltwo}{\llfragment[2]{mall}}

\newcommand{\mall}{\llfragment{mall}}

\newcommand{\mumall}{\ensuremath{\mu}\llfragment{mall}}
\newcommand{\elltwo}{\llfragment[2]{ell}}

\newcommand{\bigo}[1]{\mathcal{O}(#1)}
\newcommand{\naturalN}{\mathbb{N}}
\newcommand{\nf}[1]{\NF(#1)}

\newcommand{\card}[1]{\mathrm{card}(#1)}
\newcommand{\trame}[1]{\mathopen{|}#1\mathclose{|}}

\newcommand{\Str}{\mathtt{Str}}
\newcommand{\Bool}{\mathtt{Bool}}
\newcommand{\Nat}{\mathtt{Nat}}

\newcommand{\denot}[1]{\mathopen{\llbracket}#1\mathclose{\rrbracket}}

\newcommand{\trace}[2][]{\mathrm{Tr}_{#1}(#2)}
\newcommand{\Tr}{\mathrm{Tr}}

\newcommand{\El}{\mathsf{El}}
\newcommand{\nofo}[2]{\langle #1 \vdash #2 \rangle}

\newcommand{\Card}{\mathrm{Card}}
\newcommand{\NF}{\mathrm{NF}}
\newcommand{\CohI}{\mathsf{CohI}}
\newcommand{\CohL}{\mathsf{CohL}}
\newcommand{\CohA}{\mathsf{CohA}}
\newcommand{\HCohL}{\mathsf{HCohL}}
\newcommand{\HCohI}{\mathsf{HCohI}}

\newcommand{\powerfin}{\mathcal{P}_{\mathrm{fin}}}
\newcommand{\subsetfin}{\subseteq_{\mathrm{fin}}}
\newcommand{\hc}{\Gamma}
\newcommand{\hcns}{\Gamma^{**}}
\newcommand{\inl}{\mathsf{inl}}
\newcommand{\inr}{\mathsf{inr}}
\newcommand{\Const}{\mathrm{Const}}
\newcommand{\Val}{\mathrm{Val}}
\newcommand{\PVal}{\mathrm{PVal}}
\newcommand{\Proj}{\mathrm{Proj}}

\theoremstyle{definition}
\newtheorem{notation}[theorem]{Notation}

\usepackage{cleveref}

\begin{document}

\maketitle

\begin{abstract}
  Many applications of denotational semantics, such as higher-order model
  checking or the complexity of normalization, rely on finite semantics for
  monomorphic type systems. We exhibit such a finite semantics for a polymorphic
  purely linear language: more precisely, we show that in Girard's semantics of
  second-order linear logic using coherence spaces and normal functors, the
  denotations of multiplicative-additive formulas are finite.

  This model is also effective, in the sense that the denotations of formulas
  and proofs are computable, as we show. We also establish analogous results for
  a second-order extension of Ehrhard's hypercoherences; while finiteness holds
  for the same reason as in coherence spaces, effectivity presents additional
  difficulties.

  Finally, we discuss the applications our our work to implicit computational
  complexity in linear (or affine) logic. In view of these applications, we
  study cardinality and complexity bounds in our finite semantics.
\end{abstract}

\section{Introduction}

Polymorphism is a central topic in theoretical computer science since the sixties. 
% A breakthrough in the logical understanding of such a central concept was
A breakthrough in its logical understanding was its analysis by means of second
order quantifiers, that is the introduction of System F (also known as the
polymorphic $\lambda$-calculus) at the beginning of the seventies. This
considerable success later led Jean-Yves Girard to develop a denotational
semantics for System F \cite{girardF}, to get a deeper understanding of its
computational features. Indeed, the general goal of denotational semantics is to
give a ``mathematical'' counterpart to syntactic devices such as proofs and
programs, thus bringing to the fore their essential properties. Sometimes this
eventually results in improvements of the syntax: Linear Logic itself
\cite{girardLL} arose precisely from the denotational model introduced
in~\cite{girardF}.
% It maps the concrete 
% syntactic objects to an algebraic, geometric, categorical$\ldots$description, 
% which stresses basic invariants and, sometimes, eventually results in 
% improvements of the syntax: Linear Logic itself \cite{girardLL} arose 
% precisely from the denotational model introduced in~\cite{girardF}.

% J'aime bien la phrase mais c'est pour faire de l'optimisation
% du format deux colonnes.

But denotational semantics is not just a matter of increasing our understanding 
of programming languages, it also has direct algorithmic applications. Let us
mention:
\begin{itemize}
\item in the simply-typed lambda calculus ($\STlam$), the \emph{semantic
    evaluation} technique for complexity bounds, see Terui's recent
  paper~\cite{TeruiSemantics} and references therein;
\item in $\STlam$ extended with a fixed-point combinator, the semantic approach
  to \emph{higher-order model checking} (HOMC) advocated by Salvati and
  Walukiewicz~\cite{Salvati1,Salvati2} (see also~\cite{Aehlig,HoffmanLedent}).
\end{itemize}

The following little-known theorem illustrates both kinds of applications.
Indeed, it is a result in \emph{implicit computational complexity}: it gives a
machine-free characterization of a complexity class. At the same time, it is an
instance of the correspondence between Church encodings and automata that HOMC
generalizes to infinite trees.

\begin{theorem}[Hillebrand and Kanellakis~\cite{HillebrandKanellakis}]
  \label{thm:hk}
  The languages decided by $\STlam$ terms from Church-encoded binary strings to
  Church booleans, i.e.\ of type $\Str_\STlam[A] \to \Bool_\STlam$, are
  exactly the \emph{regular} languages.
\end{theorem}
(Here $\Str_\STlam[A] = (A \to A) \to (A \to A) \to (A \to A)$ and
$\Bool_\STlam = o \to o \to o$; $o$ is a base type, while $A$ may be chosen
depending on the language one wants to decide.)

To prove this, the main idea is to build a deterministic finite automaton (DFA)
computing the denotation of the input string. Crucially, this relies on the
existence of a \emph{finite semantics} for $\STlam$ -- such as the category of
finite sets -- which will provide the states of the DFA. In general, this
finiteness property, or finer cardinality bounds, are key to these applications.

This theorem also holds when replacing $\STlam$ by propositional linear logic,
which also admits finite semantics. In fact, Terui's solution to the complexity
of $\STlam$ normalization at fixed order~\cite{TeruiSemantics} relies on such a
semantics. As for HOMC, Grellois and Melliès have developed an approach relying
on models of linear logic~\cite{GrelloisMellies,grellois}. 

However, all of this concerns only \emph{monomorphic} type systems, for a simple
reason: as we shall soon see, \emph{System F does not admit any non-trivial
  finite semantics}. A central message of this paper is that second-order
quantification is not the only culprit here: one can also blame
\emph{non-linearity}, i.e.\ the possibility of duplicating data. What we show is
that a semantics for a purely linear language with impredicative polymorphism
can be finite:
\begin{theorem}
  \label{thm:existence-finite}
  Second-order Multiplicative-Additive Linear Logic (\malltwo) admits a
  non-trivial (i.e.\ distinguishing the two inhabitants of $1 \oplus 1$) finite
  semantics.
\end{theorem}
Recall that \malltwo is the fragment of second-order linear logic without the
\emph{exponential modalities} $\oc/\wn$ whose role is to allow\footnote{The
  possibility of duplication is expressed through the contraction rule $\oc A
  \multimap \oc A \otimes \oc A$.} a controlled amount of non-linearity.
We shall also prove the analogous property for the second-order \emph{affine}
$\lambda$-calculus.

\subsection{Some immediate consequences of finite semantics}

To illustrate the power of the above theorem, we find it instructive to explain
first the impossibility of finite semantics for System F. It is a consequence of
its ability to represent infinite data types, with definable destructors.
\begin{proposition}
  Let $\Nat_F = \forall X.\, (X \to X) \to (X \to X)$ be the type of System F
  natural numbers. Then any semantics distinguishing the two inhabitants of the
  type of booleans $\Bool_F = \forall X.\, X \to X \to X$ is injective on $\Nat$
  -- implying the latter has an infinite denotation.
\end{proposition}
\begin{proof}
  For any $n \in \naturalN$, one can define the predicate $\mathtt{eq}_n :
  \Nat_F \to \Bool_F$ which tests if its argument is equal to\footnote{We
    identify natural numbers with their Church encoding.} $n$. Thus, if $m$ and
  $n$ share the same denotation, then $\mathtt{eq}_n(n) =_\beta \mathtt{true}$
  and $\mathtt{eq}_n(m)$ have the same denotation. Since we have assumed that
  $\mathtt{true}$ and $\mathtt{false}$ have different denotations, this means
  that $\mathtt{eq}_n(m) =_\beta \mathtt{true}$, i.e.\ $m = n$.
\end{proof}

The argument above is robust enough to apply to a wide variety of situations.
For instance, the existence of a finite semantics for $\STlam$ immediately
yields by contrapositive:
\begin{theorem}[Statman]
  \label{thm:statman}
  Equality is not definable on the Church integers in $\STlam$.
\end{theorem}
% TODO find reference

Another such situation is the system \mumall: propositional \mall with fixed
points~\cite{Baelde}. Some functions analogous to $\mathtt{eq}_n$ can also be
defined in \mumall, using $\mu X.\, 1 \oplus X$ as the type of natural numbers.
This leads to a first application of our \cref{thm:existence-finite}:
\begin{theorem}\label{thm:mumall}
  There exists no faithful translation from \mumall to \malltwo.
\end{theorem}
\begin{proof}
  Else, \mumall would admit a non-trivial finite semantics, by translating it
  into \malltwo and using the semantics of \malltwo. By the argument above,
  this leads to a contradiction.
\end{proof}
While \mumall can be translated in full second-order linear logic, it was argued
that both polymorphism and exponentials were required for such a
translation~\cite[\S 2.3]{Baelde}. Our finite semantics provides a short
rigourous proof of the necessity of exponentials, and pinpoints the main reason:
\emph{\malltwo cannot represent infinite data types}.

This divide between \malltwo and \mumall is further explored in this paper in
the context of implicit complexity; this is discussed in \cref{sec:intro-reg}.

\begin{remark}
  Although $\STlam$ can represent functions on infinite data types such as
  integers or strings, this generally involves a meta-level universal
  quantification, see e.g.\ Hillebrand and Kanellakis's theorem. In \malltwo,
  which already contains quantifiers, this way of sidestepping the issue does
  not work.
\end{remark}

\subsection{Concrete models: coherence spaces and hypercoherences}

It turns out that a finite semantics of \malltwo has been lying around all along
since the birth of linear logic, though we are not aware of anyone noticing this
fact beforehand. It is none other than Girard's \emph{coherence
  spaces}~\cite{girardLL}, which he obtained as a simplification of his previous
work on \enquote{qualitative domains} in System F~\cite{girardF}, discovering
linear logic along the way.

In this model, open types (i.e.\ formulas with free variables) are represented
as \emph{normal functors}\footnote{This categorical tool had been used
  previously to give the first quantitative semantics of the
  $\lambda$-calculus~\cite{GirardNF} -- a work which is arguably one of the main
  inspirations for linear logic (linear $\lambda$-terms are interpreted in this
  model as monomials of degree 1, hence \enquote{linear}), and even differential
  linear logic~\cite{dill}.}. To prove our finiteness theorem, we introduce a
notion of normal functor of \emph{finite degree}, which is preserved by \malltwo
connectives and ensures finiteness. It is also equivalent to an asymptotic
\emph{polynomial growth} property.

Furthermore, this semantics is more concrete than its formulation using
category-theoretic machinery could suggest. Thanks to a combinatorial
presentation, we prove that it is an \emph{effective} model: the denotations of
types and terms are computable. Note that historically, effectivity was a major
motivation for using coherence spaces instead of qualitative domains
(see~\cite[Appendix~C]{girardF}, especially Remark~C.3).

We also study the interpretation of \malltwo in Ehrhard's
\emph{hypercoherences}~\cite{hypercoh}. Although it was defined as a model of
propositional linear logic, transposing the recipe of coherence spaces gives a
model of second-order linear logic, with finite denotations for \malltwo
formulas and proofs. In this semantics, effectivity stumbles upon the same issue
as in qualitative domains: roughly speaking, the presence of \enquote{$n$-ary
  coherences} for arbitrary $n \geq 2$ (while in coherence spaces, $n = 2$). We
show that despite this, the hypercoherence model can be made effective for
\malltwo, using a notion of \enquote{specification by projections}.

\subsection{Relevance to implicit computational complexity}
\label{sec:intro-reg}

The semantic developments we present here have already been applied to obtain
some results on variants of Elementary Linear Logic (ELL)~\cite{DanosJoinet}. In
this subsystem of linear logic, purely \enquote{geometric} restrictions inspired
by the theory of proof nets enforce complexity constraints, following an
approach pioneered by the characterization of polynomial time in Light Linear
Logic (LLL)~\cite{GirardELL}. Thanks to our finite semantics, we can apply to
second-order ELL some ideas from another tradition in implicit complexity,
exemplified by Hillebrand and Kanellakis's \cref{thm:hk}. (Although some
previous works on \enquote{light logics} such as LLL and ELL make use of
semantic arguments -- for instance, Statman's above-mentioned \cref{thm:statman}
has been applied to LLL in~\cite{DalLagoBaillot} -- to our knowledge these
applications have mostly consisted in proving inexpressivity results for
monomorphic systems.)
\begin{itemize}
\item In~\cite{ealreg}, we characterize \emph{regular languages} in the
  \emph{elementary affine $\lambda$-calculus}~\cite{Benedetti}; as a side
  effect, this answers an open question on a pre-existing characterization of
  polynomial time (we refer to~\cite{ealreg} for further discussion of the
  significance of this result). A crucial ingredient is the existence of a
  finite semantics of the second-order affine $\lambda$-calculus, which we shall
  prove in this paper using coherence spaces.
\item In a joint work with P.\ Pradic~\cite{sequel}, we define a class of
  queries over finite structures expressed in ELL, which lies between
  deterministic and non-deterministic \emph{logarithmic space}. (We also obtain
  a somewhat contrived exact characterization of deterministic logarithmic
  space.) This relies on some cardinality and complexity bounds on the coherence
  semantics of \malltwo, which we establish in the present paper.
\end{itemize}
As an illustration of the power of finite second-order coherence spaces, we
prove a slight variation of the first item above (regular languages) in the
setting of second-order ELL. By using some specific features of this model, we
get a shorter proof than in~\cite{ealreg}, which is very close to the proof of
\cref{thm:hk} by Hillebrand and Kanellakis.

One should note that if we were to enrich ELL with type fixpoints, then, instead
of regular languages, one would obtain a class containing at least\footnote{We
  believe that this class would be exactly P. But to adapt the P soundness
  theorem in~\cite{Baillot} to ELL, one would need to work with some notion of
  proof net for second-order (elementary) linear logic with additives; since
  this is likely to involve technical complications, we have not attempted to do
  so.} the languages decidable in polynomial time (as can be shown by adapting
the proof of polynomial time completeness in~\cite{Baillot}; see also the
discussion in~\cite{ealreg}). This gives a quantitative manifestation of the
expressivity gap between \mumall and \malltwo mentioned earlier in this
introduction.

% TODO mention transducers (when preparing final version?)

\subsection{Plan of the paper}

We first recall the second-order coherence space model and prove its finiteness
and effectivity for \malltwo in \cref{sec:coh}. This is followed by a short
discussion, in \cref{sec:affine}, on the adaptation of this finite semantics to
the second-order affine $\lambda$-calculus. We define the second-order extension
of the hypercoherence model in \cref{sec:hypercoh}, and show how to make it
effective.

\section{Finite and effective second-order coherence spaces}
\label{sec:coh}

Before we tackle the question of finiteness, we must first recall Girard's model
of second-order linear logic in coherence spaces. Since this model is not very
well-known, the first two subsections will be expository with no new results.
Finiteness is shown in \cref{sec:coh-finiteness}, and effectivity in
\cref{sec:combinatorial}. The omitted proofs of this section are in
\cref{sec:appendix-coh}.

\subparagraph{Syntax}

The formulas of \emph{second-order linear logic} are given by the grammar
\begin{equation*}
A,B := X \mid X^\perp \mid 1\mid \bot \mid A\otimes B\mid A\parr B
\mid  0\mid \top \mid A\oplus B\mid A\with B
\mid \forall X.\, A\mid \exists X.\, B \mid \oc A \mid \wn A 
\end{equation*}
where $X$ belongs to a fixed countable set of variables. \emph{Propositional}
linear logic is the fragment made of formulas without quantifiers
$\exists/\forall$, while \malltwo is the fragment without the the exponential
modalities $\oc$/$\wn$.

The involutive \emph{linear negation} $(-)^\perp$ is defined inductively on
formulas by the rules in~\cref{fig:mall2-macros}. It is used to define linear
implication as $A \multimap B := A^\perp \parr B$.

Since we work in pre-existing semantics, we do not need to formally define the
notion of denotational model of linear logic; we refer to~\cite{MelliesPanorama}
for an extensive survey of this topic. We will not be need to work with some
precise proof system -- e.g.\ sequent calculus -- either, except in
\cref{sec:regular}.

\begin{figure}
\centering
\[
\begin{array}{lcl!\qquad lcl !\qquad lcl}
1^\bot &:=& \bot &
\bot^\bot &:=& 1 &
(\exists X.\, A)^\bot &:=& \forall X.\,A^\bot
\\
(A \otimes B)^\bot &:=& A^\bot \parr B^\bot &
(A \parr B)^\bot &:=& A^\bot \otimes B^\bot & (\forall X.\, A)^\bot &:=& \exists X.\,A^\bot\\
0^\bot &:=& \top &
\top^\bot &:=& 0 & (\oc A)^\bot &:=& \wn A^\bot \\
(A \oplus B)^\bot &:=& A^\bot \with B^\bot &
(A \with B)^\bot &:=& A^\bot \oplus B^\bot &
(\wn A)^\bot &:=& \oc A^\bot
\end{array}
\]
\caption{Duality for formulas of linear logic.}
\label{fig:mall2-macros}
\end{figure}

\subparagraph{Coherence spaces (propositional case)}

Recall that a \emph{coherence space} is an undirected (reflexive) graph, i.e.\ a
pair $X=(\trame{X},\coh_{X})$ of a set $\trame{X}$ -- the \emph{web} of $X$ --
and a symmetric and reflexive relation $\coh_{X}\;\subseteq
\trame{X}\times\trame{X}$ -- its \emph{coherence relation}. A subset $c
\subseteq \trame{A}$ is a \emph{clique} of $A$ if its points are pairwise
coherent for $\coh_A$; in this case we write $c \sqsubset A$.

The operations $(-)^\perp, \otimes, \oplus, \oc$ are defined on coherence spaces
as follows (cf.~\cite{LLSS}):
\begin{itemize}
\item $\trame{X^\perp} = \trame{X}$ and $x \coh_{X^\perp} x' \iff x
  \not\coh_{X} x' \lor x = x'$ (complement graph)
\item $\trame{X \otimes Y} = \trame{X} \times \trame{Y}$ and $(x,y) \coh_{X
    \otimes Y} (x',y') \iff x \coh_X x' \land y \coh_Y y'$
\item $\trame{X \oplus Y} = \trame{X} \uplus \trame{Y}$ and $z \coh_{A \oplus B}
  z' \iff \exists Z \in \{X,Y\} : (z \in Z \land z' \in Z \land z \coh_Z z')$
\item $\trame{\oc X} = \{\text{cliques of $X$}\}$ and $c \coh_{\oc X} c' \iff c
  \cup c' \sqsubset X$ (\enquote{set-based} exponential)
\end{itemize}
Furthermore, the multiplicative units $1,\bot$ are interpreted as the singleton
space, and the additive units $\top,0$ as the empty space. This is enough to
define inductively the denotation $\denot{A}$ of a formula $A$ in propositional
linear logic, given an assignment of the free variables of $A$ to coherence
spaces.

A proof $\pi : A$ is interpreted in the coherence space model as a clique
$\denot{\pi} \sqsubset \denot{A}$. In terms of categorical semantics, the model
is given as the category $\CohL$:
\begin{itemize}
\item whose objects are coherence spaces;
\item whose morphisms between $X$ and $Y$ are cliques in $X \multimap Y$,
  composition being relational composition (this is meaningful since $\trame{X
    \multimap Y} = \trame{X} \times \trame{Y}$).
\end{itemize}

\subsection{Functors on embeddings and uniform families}

In order to interpret second-order quantification, we want to give a first-class
status to the map $\{\text{assignments for variables in $A$}\} \to
\{\text{possible values for $\denot{A}$}\}$ when $A$ is an open type (i.e.\ a
formula with free variables). The first idea that comes to mind is to consider
it as a functor on $\CohL$. But it stumbles on the fact that while the binary
connectives are \emph{covariant} bifunctors on $\CohL$, linear negation is a
\emph{contravariant} endofunctor. Instead, Girard's idea is to work in a
``category of embeddings'' (see~\cite{Coquand}) to make negation covariant.

\begin{definition}
  An \emph{embedding} of a coherence space $X$ into a coherence space $Y$ is an
  injection $\iota : \trame{X} \to \trame{Y}$ such that $x \coh_X x'
  \Leftrightarrow \iota(x) \coh_Y \iota(x')$. We write $\iota : X
  \hookrightarrow Y$.
  
  The category $\CohI$ has as objects coherence spaces, and as morphisms the
  embeddings.
\end{definition}
\begin{proposition}
  $(-)^\perp$ is a \emph{covariant} endofunctor of $\CohI$.
\end{proposition}
\begin{proof}
  If $X$ is an induced subgraph of $Y$, then $X^\perp$ is an induced subgraph of
  $Y^\perp$.
\end{proof}
\begin{remark}
  In the same vein, the graph $\{(x,\iota(x)) \mid x \in \trame{X} \}$ of an
  embedding $\iota : X \hookrightarrow Y$ is a clique both in $X \multimap Y$ and in
  $X^\perp \multimap Y^\perp$.
\end{remark}

Let us say, provisionally, that functors $F : \CohI^n \to \CohI$ are our
semantical counterpart of open formulas with $n$ variables. A proof of such a
formula should be a family of cliques $c_{X_1, \ldots, X_n} \sqsubset F(X_1,
\ldots, X_n)$, \enquote{uniform} in some way. The following notion of uniformity
has been called the \enquote{mutilation property} by Girard~\cite{girardF}:
\begin{definition}
  A family $c_{X_1, \ldots, X_n} \sqsubset F(X_1, \ldots, X_n)$ is
  \emph{uniform} if for any embeddings $\iota_i : X_i \hookrightarrow Y_i$ ($i
  \in \{1, \ldots, n\}$), $c_{X_1, \ldots, X_n} = F(\iota_1, \ldots,
  \iota_n)^{-1}(c_{Y_1, \ldots, Y_n})$.
\end{definition}

\begin{remark}
  \label{rem:composition}
  At this point we have to point out a subtlety of the coherence space model: it
  is not a priori obvious that uniformity is closed under composition, in other
  words, that the pointwise composition of a uniform clique family for $F
  \multimap G$ with a uniform clique family for $G \multimap H$ is uniform for
  $F \multimap H$. Indeed, taking $n = 1$ for simplicity, the uniformity
  condition for a family $f_X \sqsubset F(X) \multimap G(X)$ seen as a family of
  morphisms is expressed as a diagram
  \begin{center}
    \begin{tikzcd}[column sep=large]
      F(Y) \ar[r, "f_Y"] & G(Y) \ar[d, "G(\iota)^-"] \\
      F(X) \ar[u, "F(\iota)^+"] \ar[r, "f_X"] & G(X)
    \end{tikzcd}
  \end{center}
  in $\CohL$, where $j^+ = \{(x,j(x)) \mid x \in A\} \sqsubset A \multimap B$
  and $j^- = \{(j(x),x) \mid x \in A\} \sqsubset B \multimap A$ for $j : A
  \hookrightarrow B$. Such diagrams cannot be \enquote{formally} pasted
  horizontally.

  It turns out that uniform clique families \emph{do} compose, but this issue is
  non-trivial and was overlooked in Girard's papers. A proof is part of the
  folklore and often credited to Eugenio Moggi. % TODO: proof in hypercoh case
\end{remark}
\begin{remark}
  In the relational semantics of linear logic, in which objects are sets and
  morphisms are relations, one could define an analogous notion of uniform
  subset family for a functor on the category of injections. But then, for the
  uniform families
  \[ c_S = \{*\} \times S \subseteq \denot{1 \multimap X}_{X \mapsto S} \qquad
    c'_S =  S \times \{*\} \subseteq \denot{X \multimap 1}_{X \mapsto S} \]
  the composition $c'_S \circ c_S$ is \emph{not uniform}: it is equal to
  $\{(*,*)\}$ if $S \neq \varnothing$, and $\varnothing$ if $S = \varnothing$.

  This issue with the second-order relational model is known and has been
  investigated by A.\ Bac-Bruasse (whose PhD thesis in French~\cite{Bac} is the
  main reference on the subject to our knowledge), T.\ Ehrhard and C.\ Tasson.
  What makes composition work in coherence spaces is the domain-theoretic
  \emph{stability}\footnote{Stability is indeed a recurring pattern here: the
    uniformity condition is reminiscent of Berry's stable order between
    functions on domains, and the preservation of pullbacks in normal functors
    (next section) is a categorification of stability.} property of morphisms.
  This is our reason for working with coherence spaces instead of the simpler
  relational model.
\end{remark}

\subsection{Normal functors}
\label{sec:normal-functors}

The next goal is to interpret quantifiers. Let us look at the example of the
formula $X \multimap X$, which admits a proof $\pi$ such that $\denot{\pi}_{X
  \mapsto S} = \{(s,s) \mid s \in S\}$ for any coherence space $S$. This uniform
family should correspond to a clique in some space $\denot{\forall X.\, X
  \multimap X}$. The idea is to take some kind of \enquote{patterns with bound
  variables} as the points of this coherence space. Typically, $\{(s,s) \mid s
\in S\}$ should correspond to the single pattern $\nofo{x}{(x,x)}$ with a
bound variable $x$ -- and thus to the clique $\{\nofo{x}{(x,x)}\} \sqsubset
\denot{\forall X.\, X \multimap X}$. Observe:
\begin{itemize}
\item that the substitution of the variable $x$ by $s \in S$ in $(x,x)$
  corresponds to the functoriality of $\denot{X \multimap X} : \CohI \to \CohI$
  with respect to the embedding $\iota : \{x\} \hookrightarrow S$ such that
  $\iota(x) = s$;
\item that $(x,x) \in F(\{x\})$, and $\{x\}$ is \enquote{minimal} or
  \enquote{initial} in the sense that any other $(s,s) \in F(S)$ is an image of
  $(x,x)$ via an embedding $\{x\} \hookrightarrow S$.
\end{itemize}
We thus arrive at the idea that \enquote{patterns with bound variables} should
correspond to \enquote{minimal} spaces. To guarantee their existence, we need to
put an additional condition on our functors $\CohI^n \to \CohI$. This is why
Girard interprets open types by \emph{normal functors}.

\begin{definition}
  A functor is \emph{normal} if it preserves filtered colimits and finite
  pullbacks.
\end{definition}

The name comes from Girard's \emph{normal form theorem}:

\begin{theorem}
  Let $F : \CohI^n \to \CohI$ be a functor, $\trame{F}$ be the covariant
  presheaf obtained by taking the web, and $\El(\trame{F})$ be its category of
  elements.

  $F$ is normal if and only if, for any object $\vec{X}$ in $\CohI^n$ and point
  $x \in \trame{F(\vec{X})}$, the slice category $\El(\trame{F})/(\vec{X},x)$ admits a
  \emph{finite initial} object $(\vec{X'},x')$.

  In this case, $(\vec{X'},x')$ is initial in its own slice category. We call an
  object of $\El(\trame{F})$ enjoying this property a \emph{normal form}.
\end{theorem}

\begin{remark}
  It is worth noting that this characterisation is one of many similar results.
  For instance, Joyal's analytic functors have a \emph{weak} finite normal form
  property (i.e. where the initial elements are only weakly initial), a
  variation corresponding to preservation of weak pullbacks and filtered
  colimits. Similarly, Kock's characterisation of polynomial functors states
  that preservation of wide pullbacks is equivalent to the existence of normal
  forms (though not finite); in fact, Girard's normal functors correspond to
  Kock's \emph{finitary polynomial functors}. See the discussion in~\cite[\S
  1.18--1.21]{KockGambino}.
\end{remark}

\begin{definition}
  Let $F$ be a normal functor. We define $\nf{F}$ to be its set of isomorphism
  classes of normal forms (for isomorphisms in $\El(\trame{F})$).
\end{definition}

\begin{notation}
  We use the notation $\nofo{\vec{X}}{x}$ for normal forms $(\vec{X},x) \in \NF(F)$.
  Alternatively, if $\vec{X} = (X_1, \ldots, X_n)$, we may write
  $\nofo{X_1, \ldots, X_n}{x}$.
\end{notation}

The set $\NF(F)$ summarizes in a way all the webs $|F(\vec{X}|)$ of the
instantiations of $F$, as the proposition below shows. As for the uniform
families of cliques of $F$, they are summarized by the \emph{trace} of $F$, a
coherence space built from $\NF(F)$.

\begin{proposition}
  \label{prop:instantiation-space}
  Let $F : \CohI^n \to \CohI$ be a normal functor and $\vec{X} = (X_1, \ldots,
  X_n)$. Then $\trame{F(\vec{X})} = \{ F(\iota_1, \ldots, \iota_n)(y) \mid
  \nofo{\vec{Y}}{y} \in \NF(F),\, \iota_i : Y_i \hookrightarrow
  X_i\;\text{for}\; i \in \{1 \ldots n\}\}$.
\end{proposition}

\begin{definition}
  \label{def:trace}
  Let $F : \CohI^n \to \CohI$ be a normal functor. We endow $\nf{F}$ with a
  non-reflexive coherence relation: $\nofo{\vec{X}}{x} \coh_{\nf{F}}
  \nofo{\vec{Y}}{y}$ when for all $n$-tuples $\vec{Z}$ and all embeddings
  $\iota_{\vec{X},i} : X_i \hookrightarrow Z_i$ and $\iota_{\vec{Y},i}: Y_i
  \hookrightarrow Z_i$, $F(\iota_{\vec{X},1},\ldots,\iota_{\vec{X},n})(x)
  \coh_{F(\vec{Z})} F(\iota_{\vec{Y},1},\ldots,\iota_{\vec{Y},n})(y)$.

  The \emph{trace} $\trace{F}$ is defined as the coherence space made of the
  self-coherent normal forms of $F$, equipped with the coherence relation above.
\end{definition}

\begin{proposition}
  \label{prop:instantiation-clique}
  There is a bijection between the cliques $c \sqsubset \Tr(F)$ and the uniform
  families of cliques $c_{\vec{X}} \sqsubset F(\vec{X})$ for a normal functor $F : \CohI^n
  \to \CohI$, given by
  \[ c_X = \{ F(\iota_1, \ldots, \iota_n)(y) \mid \nofo{\vec{Y}}{y} \in c,
    \iota_i : Y_i \hookrightarrow X_i\;\text{for}\; i \in \{1 \ldots n\} \} \]
\end{proposition}

This leads to the interpretation of quantifiers. One interprets inductively any
formula of second-order linear logic $A$ with $n$ free variables into a normal
functor $\denot{A} : \CohI^n \to \CohI$: the connectives $\oplus, \otimes,
(-)^\perp$ extend to \enquote{pointwise} operations on normal functors, and the
case $A = \forall X.\, B$ is handled by the operation introduced below.

\begin{proposition}
  Let $F$ be a normal functor $F : \CohI^{n+1}\rightarrow \CohI$. The map on
  objects $\forall(F)(X_1, \ldots, X_n) = \trace{F(X_1,\dots,X_n,-)}$ extends to
  a normal functor $\forall(F) : \CohI^n \to \CohI$.
\end{proposition}

% TODO preuve

\subsection{Ensuring finiteness: normal functors of finite degree}
\label{sec:coh-finiteness}

We now come to our technical contributions, having just finished the exposition
of Girard's model. This section introduces a notion of \emph{degree} of a normal
functor, which will witness the finiteness of the interpretation of \malltwo.

\begin{definition}
  \label{def:finite-degree}
Let $F:\CohI^{n}\rightarrow\CohI$ be a normal functor. We define
the degree of $F$ as:
\[
\deg(F) = \sup\;\{\card{\trame{X_i}} \mid \nofo{X_1, \ldots, X_n}{x} \in \nf{F}, i\in\{1,\dots,n\}\}. 
\]
We say $F$ is \emph{finite} if it preserves finiteness of cardinality 
and is of finite degree.
\end{definition}

Note that a normal functor may have finite but unbounded normal forms, so that
its degree is in fact infinite. Typically, this is the case for the exponential
modalities, which explains why the model is not finite for full second-order
linear logic. We now give two characterizations of finite normal functors.

\begin{proposition}
  \label{prop:nf-finite}
  A normal functor $F$ is finite if and only if $\NF(F)$ is finite.
\end{proposition}

\begin{theorem}[Finiteness = polynomial growth]
  \label{thm:polynomialgrowth}
  Let $F : \CohI\rightarrow \CohI$ be a normal functor. There exists $d\in
  \naturalN$ s.t. $\card{\trame{F(X)}} = \bigo{\card{\trame{X}^{d}}}$ if and
  only if $F$ is a \emph{finite} normal functor. In that case, $\deg{F}$ is the
  smallest such $d$.
\end{theorem}

For applications such as the one in \cref{sec:regular}, the relevant notion of
\enquote{finite semantics} is a model with finite sets of morphisms. Finite
normal functors achieve this requirement.
\begin{proposition}
  \label{prop:finite-cliques}
  A finite normal functor has finitely many uniform families of cliques.
\end{proposition}
\begin{proof}
  By \cref{prop:nf-finite} together with \cref{prop:instantiation-clique}, since
  $\trame{\Tr(F)} \subseteq \NF(F)$.
\end{proof}

To obtain a finite semantics, our goal is therefore to show that inside the
model of coherence spaces and normal functors, the finite ones constitute a
submodel of \malltwo.

\begin{proposition} 
  \label{prop:degree-mallzero}
  If $F$ and $G$ normal functors in $\CohI^{n}\rightarrow\CohI$, then
  $\deg{F^{\bot}} = \deg{F}$, $\deg{F\otimes G} = \deg{F} + \deg{G}$, and
  $\deg{F\oplus G} = \max\{\deg{F}, \deg{G}\}$.
\end{proposition}

\begin{proposition}
  \label{prop:degree-forall}
  For any normal functor $F: \CohI^{n+1}\rightarrow \CohI$,
  $\deg{\forall(F)}\leqslant \deg{F}$.
\end{proposition}

\begin{theorem}
  \label{thm:closure-malltwo}
  Finite normal functors are closed under \malltwo connectives.
\end{theorem}
\begin{proof}%[Proof of \cref{thm:closure-malltwo}]
  We still need to show that if $F,G : \CohI^n \to \CohI$ are finite normal
  functors, then $F \oplus G, F \otimes G$ and $\forall(F)$ preserve finiteness
  of cardinality. This is immediate for the first two, and the latter reduces to
  the case $n = 1$: we must show that $\Tr(F)$ is finite. This follows from
  \cref{prop:nf-finite} since $\trame{\Tr(F)} \subseteq \NF(F)$.
\end{proof}

The above results, together with \cref{prop:finite-cliques}, entails the
\cref{thm:existence-finite} claimed in the introduction. We can be a bit more
precise:

\begin{corollary}
  \label{cor:fincoh}
  Let $A$ be a formula of second-order linear logic. Suppose that in all
  subformulas of $A$ of the form $\oc B$ or $\wn B$, any type variable in $B$ is
  bound by a quantifier in $B$. Then $\denot{A}$ is a finite normal functor. In
  particular:
  \begin{itemize}
  \item this applies when $A$ is a \malltwo formula;
  \item when $A$ is closed, $\denot{A}$ is a finite coherence space.
  \end{itemize}
\end{corollary}

\subsection{Effectivity properties via a combinatorial description}
\label{sec:combinatorial}

We are now ready to revisit the example outlined at the start of
\cref{sec:normal-functors}, and discuss in more generality the
\enquote{combinatorial} or \enquote{syntactic} presentation of the
\malltwo-definable coherence spaces. Our exposition here is inspired by the
description of normal functors over the category of sets and injections
in~\cite[\S IV.5]{Bac}.

The idea is to see the webs $\trame{X_i}$ in a normal form $\nofo{X_1, \ldots,
  X_n}{x}$ as sets of \emph{bound variables} in $x$. Recall that these normal
forms are considered up to isomorphism in a category of elements
$\El(\trame{F})$; these isomorphisms should be understood as $\alpha$-renamings.
The initiality condition on normal forms means that all the variables in the
$\trame{X_i}$ appear free in $x$ -- otherwise, one could take a smaller $X'_i$.
Note that the coherence spaces $X_i$ specify not only which variables are bound,
but also the coherence relation between them.

In turn, this $x$ is a syntax tree with binders -- indeed the interpretation of
quantifiers uses (unary) normal forms. The \malltwo connectives induce a grammar
of terms
\[ x ::= a \in \mathrm{Var} \mid (x,x) \mid \mathrm{inl}(x) \mid \mathrm{inr}(x)
  \mid \nofo{X}{x} \]
where $\trame{X} \subset \mathrm{Var}$. The functorial action of a
\malltwo-definable functor $F$ on embeddings then corresponds to substitution --
indeed an embedding $\iota_i : X_i \hookrightarrow Y_i$ is an assignment of
variables.

The shape of the term is in fact heavily constrained by the \malltwo formula
which $F$ interprets. With this point of view, one sees that $\deg(F)$ is the
maximum number of leaves which a syntax tree in $\nf{F}$ can have.

With such a concrete description it becomes easier to see how one can compute
operations on these variable types and cliques. For instance:

\begin{proposition}
  \label{prop:cohnf-decidable}
  For any \malltwo-definable functor $F$, the non-reflexive coherence relation
  on $\nf{F}$ (\cref{def:trace}) is decidable.
\end{proposition}

This may be used to establish the \emph{effectivity} of our finite semantics of
\malltwo:

\begin{theorem}\label{thm:effectivity}
  The function sending a \malltwo formula $A$ to $\Tr(\denot{A})$ is computable.
  Futhermore, the function taking a formula $A$ and a proof $\pi : A$ as input
  and returning the clique of $\Tr(\denot{A})$ corresponding to the uniform
  family $\denot{\pi}(\vec{X}) \sqsubset \denot{A}(\vec{X})$ is computable.
\end{theorem}

% TODO really last section??

\begin{theorem}\label{thm:logspace}
  Let $F : \CohI^n \to \CohI$ be a fixed $\malltwo$-definable finite normal
  functor.
  \begin{itemize}
  \item The function $\vec{X} \mapsto F(\vec{X})$ is computable in
    \emph{logarithmic space}.
  \item For any $c \sqsubset \Tr(F)$, the function $\vec{X} \mapsto c_{\vec{X}}
    \sqsubset F(\vec{X})$ is computable in \emph{logarithmic space}.
  \end{itemize}
\end{theorem}

Since the output of a logarithmic space algorithm has polynomial size, this is
consistent with \cref{thm:polynomialgrowth}.

\section{A finite semantics of the second-order affine $\lambda$-calculus}
\label{sec:affine}

Let us discuss briefly how to adapt our finite coherence semantics from \malltwo
to its intuitionistic affine variant. Indeed, the existence of a finite
semantics for this variant is used in one of the applications to implicit
complexity mentioned in the introduction~\cite{ealreg}.

The starting point is to give a semantics of the propositional affine
$\lambda$-calculus. To do so, the natural idea is to use the well-known notion
of affine function space between coherence spaces: $X \affimplies Y = (X
\multimap Y) \with Y$, so that every clique $c \sqsubset X \affimplies Y$ can be
written uniquely as the disjoint union $c = c_1 \sqcup c_2$ of a linear function
from $X$ to $Y$ ($c_1 \sqsubset X \multimap Y$) and a constant part ($c_2
\sqsubset Y$). Composition is defined as\footnote{Fore $r \subseteq S \times T$
  and $s \subseteq S$, $r(s)$ denotes $\{y \mid \exists x \in s : (x,y) \in r
  \}$; this generalizes function composition and can be seen as a degenerate
  case of relational composition.} $c' \circ_{\mathrm{aff}} c = (c'_1 \circ c_1)
\sqcup (c'_2 \cup c'_1(c_2))$ for $c \sqsubset X \affimplies Y$ and $c'
\sqsubset Y \affimplies Z$. Coherence spaces and affine functions form a
category $\CohA$ which can also be seen as the Kleisli category for the comonad
$(- \with 1)$, via the isomorphism $(X \multimap Y) \with Y \cong (X \with 1)
\multimap Y$.

One issue is that this category does not quite fit into the definition of a
\enquote{symmetric monoidal closed category with terminal unit}, which is the
usual notion of denotational model for the affine $\lambda$-calculus. The reason
is that it does not admit a monoidal product $*$ such that, for any
coherence space $B$, $(- * B)$ is left adjoint to $(B \affimplies -)$. Let
us review the two main candidates:
\begin{itemize}
\item $* = \otimes$ fails: morally, an affine function from $A \otimes B$
  can either use both $A$ and $B$ or none of them, but not one out of the two;
\item $A * B = (A \with 1) \otimes (B \with 1)$ almost works, except for the
  fact that it is not associative\footnote{A similar problem afflicts the
    treatment of intuitionistic disjunction in domain-theoretic models.}.
\end{itemize}
To sidestep this issue, we do not attempt to interpret an affine tensor product
-- it does not exist anyway in the syntax of the affine $\lambda$-calculus (in
practice, one uses the second-order encoding\footnote{Whose interpretation in
  our coherence space semantics coincides with $(A \with 1) \otimes (B \with
  1)$.} $A \odot B = \forall X.\, (A \affimplies B \affimplies X) \affimplies
X$). Instead, we give a semantics in a \enquote{closed multicategory}. In the
absence of a well-established theory of multicategorical models of typed
$\lambda$-calculi and their extension with second-order quantification, we
merely give a concrete interpretation of second-order affine $\lambda$-terms.

An affine type in the grammar $A, B ::= X \mid A \affimplies B \mid \forall X.\,
A$ with $n$ type variables is interpreted as a normal functor $\CohI^n \to
\CohI$ by induction:
\[ \denot{A \affimplies B} = (A \with 1) \multimap B \qquad \denot{\forall X.\,
    A} = \forall(\denot{A}) \]
An affine $\lambda$-term $t : B$ with $m$ free variables $x_1 : A_1, \ldots, x_m
: A_m$ is mapped to a uniform family $\denot{t}(\vec{X}) \sqsubset
(\denot{A_1}(\vec{X}) \with 1) \otimes \ldots \otimes (\denot{A_m}(\vec{X})
\with 1) \multimap \denot{B}(\vec{X})$. The details are unsurpring and are given
in \cref{sec:appendix-affine}.
\begin{remark}
  One could try instead to use the Eilenberg--Moore category of coalgebras for
  the comonad $(- \with 1)$ instead of its Kleisli category. But in that case,
  even though the monoidal product can be defined, the function space cannot.
\end{remark}

\section{Hypercoherences}
\label{sec:hypercoh}

The remainder of the paper is dedicated to applying the ideas we developed in
coherence spaces to Ehrhard's \emph{hypercoherence}~\cite{hypercoh} model of
linear logic. To our knowledge, the treatment of quantifiers in hypercoherences
has not appeared in the literature, but it is easily constructed by analogy with
coherence spaces, as we do in \cref{sec:hypercoh-basic}. As mentioned in the
introduction, obtaining an effective model will be harder than in the case of
coherence spaces, and this is where most of our energy will be spent
(\cref{sec:hypercoh-effective}). Some proofs are in
\cref{sec:appendix-hypercoh}.

\begin{notation}
  Given a set $S$, we write $\powerfin(S)$ (resp.\ $\powerfin^*(S)$) for the set
  of finite (resp.\ finite non-empty) subsets of $S$. An alternative notation
  for $S' \in \powerfin(S)$ (resp.\ $S' \in \powerfin^*(S)$) is $S' \subsetfin
  S$ (resp.\ $S' \subsetfin^* S$).
\end{notation}

Recall that a \emph{hypercoherence} $X$ is a pair $(|X|,\hc(X))$ where $|X|$ is
a set and $\hc(X) \subseteq \powerfin^*(S)$ contains all singletons ($\forall x
\in |X|,\, \{x\} \in \hc(X)$). Equivalently, one could specify a hypercoherence
$X$ by the data of $|X|$ and $\hcns(X) = \hc(X) \setminus \{\{x\} \mid x \in
|X|\}$. A subset $c$ of $|X|$ is a \emph{clique} of $X$ (as in coherence spaces,
this is denoted $c \sqsubset X$) if $\powerfin^*(c) \subseteq \hc(X)$.

Hypercoherences support the following operations:
\begin{itemize}
\item $X^\perp = (|X|, \powerfin^*(X) \setminus \hcns(X))$
\item $X \otimes Y = (|X| \times |Y|, \{ S \mid \pi_1(S) \in \hc(X) \land
  \pi_2(S) \in \hc(Y) \})$
\item $X \oplus Y = (|X| + |Y|, \inl(\hc(X)) \cup \inr(\hc(Y))$
\end{itemize}
As in coherence spaces, one builds a category $\HCohL$ whose objects are
hypercoherences and whose morphisms between $X$ and $Y$ are the cliques of $X
\multimap Y = (X \otimes Y^\perp)^\perp$. These morphisms are composed by
relational composition.

\subsection{Variable and second-order hypercoherences}
\label{sec:hypercoh-basic}

We follow the recipe of coherence spaces to extend this model to \malltwo.

An \emph{embedding} $X \hookrightarrow Y$ of hypercoherences is an injection
$|X| \hookrightarrow |Y|$ which preserves both coherence and incoherence of
subsets (equivalently, the graph of the injection is both a morphism $X
\multimap Y$ and a morphism $X^\perp \multimap Y^\perp$). (Similarly, if $S
\subseteq |X|$, the sub-hypercoherence of $X$ \emph{induced by $S$} is $(S,
\powerfin^*(S) \cap \hc(X))$.)

\begin{definition}
  The category $\HCohI$ has hypercoherences as objects, and embeddings as
  morphisms. A $n$-parameter \emph{variable hypercoherence} is a normal functor
  $\HCohI^n \to \HCohI$.

  A variable hypercoherence $F$ is \emph{(weakly) finite} if $\NF(F)$ is finite
  -- as in coherence spaces, $\NF(F)$ is defined from the presheaf $|F|$. Its
  \emph{degree} $\deg(F)$ is defined as in \cref{def:finite-degree}.
\end{definition}
\begin{definition}
  Let $F : \HCohI^n \to \HCohI$ be a variable hypercoherence. A family $c_{X_1,
    \ldots, X_n} \sqsubset F(X_1, \ldots, X_n)$ is \emph{uniform} if for any
  embeddings $\iota_i : X_i \hookrightarrow Y_i$ ($i \in \{1, \ldots, n\}$),
  $c_{X_1, \ldots, X_n} = F(\iota_1, \ldots, \iota_n)^{-1}(c_{Y_1, \ldots,
    Y_n})$.
\end{definition}
\begin{proposition}
  For all $n \in \mathbb{N}$, the $n$-parameter variable hypercoherences and
  uniform families of cliques form a category $\HCohL(n)$ which is a model of
  propositional \mall.
\end{proposition}
The $n=0$ case is $\HCohL(0) = \HCohL$. As we saw in \cref{rem:composition}, it
is not quite trivial to show that $\HCohL(n)$ is a category, but since the
morphisms in $\HCohL$ enjoy the domain-theoretic \emph{stability} property, the
proof used in the case of coherence spaces applies \emph{mutatis mutandis}.

Now we wish to define a notion of trace $\Tr(F)$ of a normal functor $F$ to
interpret $\forall$. The fundamental thing to ensure is the adjunction (usual in
categorical semantics)
\[ \HCohL(n)(A, \forall(F)) \cong \HCohL(n+1)(\Const(A),F) \] where $\Const$ is
a \enquote{weakening} map sending a $n$-parameter hypercoherence to a
$(n+1)$-parameter one. This reduces to the case\footnote{If one were to
  specialize this further to $\HCohL(0)(1, \Tr(F)) \cong \HCohL(1)(1,F)$ --
  i.e.\ the cliques of $\Tr(F)$ are in bijection with the variable cliques of
  $F$ -- then it would not determine $\Tr(F)$ uniquely, unlike the case of
  coherence spaces: in general, the structure of a hypercoherence cannot be
  fully recovered from its domain of cliques. That said, the adjunction defining
  $\Tr(-)$ can be derived from this bijection between cliques together with
  $\Tr(\Const(A) \multimap F) \cong A \multimap \Tr(F)$ (i.e.\ commutation
  between $\forall$ and $\parr$).} $n=0$: $\HCohL(A, \Tr(F)) \cong
\HCohL(1)(\Const(A),F)$. So $\Tr$, being an adjoint functor, is unique up to
natural isomorphism\footnote{An isomorphism $X \cong Y$ in $\HCohL$ is just a
  bijection from $|X|$ to $|Y|$ sending $\hc(X)$ to $\hc(Y)$.} (if it exists);
we can just state the definition -- mimicking coherence spaces -- and check that
the adjunction holds.

\begin{definition}
  For $x \in |F(X)|$, $\NF(x;F)$ denotes (the isomorphism class of) the normal
  form of $x$. For $S \subseteq |F(X)|$, $\NF(S;F)$ is the direct image of $S$
  by the function $\NF(-;F)$.
\end{definition}

\begin{definition}
  Let $F$ be a one-parameter variable hypercoherence. $N \subsetfin^* \NF(F)$ is
  said to be coherent when for any $X$ and any $S \subsetfin^* |F(X)|$, if
  $\NF(S;F) = N$ then $S \in \hc(F(X))$. The set of coherent sets of normal
  forms is denoted by $\hc\NF(F)$.

  The \emph{trace} of $F$ (notation: $\Tr(F)$) is the hypercoherence defined by
  \[ |\Tr(F)| = \{x \in \NF(F) \mid \{x\} \in \hc\NF(F)\}\quad \hc(\Tr(F)) =
    \powerfin^*(\Tr(F)) \cap \hc\NF(F)\]
  and from this we define $\forall(G) : \vec{X} \mapsto \Tr(G(\vec{X},-))$ for
  $G : \HCohI^{n+1} \to \HCohI$.
\end{definition}

One can then routinely check that the adjunction holds.
% TODO: faire en appendix ?
We therefore conclude:
\begin{theorem}
  Variable hypercoherences and uniform families of cliques form a semantics of
  second-order linear logic. Furthermore, by restricting to weakly finite
  variable hypercoherences, we obtain a finite semantics of \malltwo.
\end{theorem}

\subsection{Strong finiteness and effectivity}
\label{sec:hypercoh-effective}

Unfortunately, the model of weakly finite variable hypercoherences is \emph{not
  effective}. Let us give an exemple: let $f : \mathbb{N} \to \{0,1\}$ be any
function, and
\[ F_f(X) = (|X|, \{S \subsetfin^* |X| \mid f(\Card(S)) = 1 \}) \]
This map on objects can be extended to a functor $F_f : \HCohI \to \HCohI$
which is in fact a weakly finite variable hypercoherence. But if $f$ is
uncomputable, then $X \mapsto \hc(F_f(X))$ also is.

We are therefore seeking an additional condition on variable hypercoherences
which would both guarantee effectivity and be preserved by \malltwo connectives.
More precisely, our goal is to exhibit a class of variable hypercoherences $F$
such that $\Gamma(F)$ can be described canonically by some finite data -- just
as $\NF(F)$ fulfills this role for $|F|$ when $F$ is weakly finite. This is the
purpose of the following defintions.

\begin{definition}
  Let $F$ be a $n$-parameter variable hypercoherence and $1 \leq k \leq n$.

  For any $P \subseteq \NF(F)$, a \emph{projection from $P$} on the $k$-th
  parameter is a dependent function
  \[ f : (\nofo{X_1, \ldots, X_n}{x} \in P) \to |X_k|\]
  i.e.\ it is a function $f$ defined on $P$ such that $f(\nofo{X_1, \ldots,
    X_n}{x}) \in |X_k|$. The set of projections from $P$ on the $k$-th parameter
  is written $\Proj_k(P)$.

  Any $f \in \Proj_k(P)$ induces a family of functions indexed by
  hypercoherences $Y_1, \ldots, Y_n$
  \[f^F_{Y_1, \ldots, Y_n} : \{y \in |F(Y_1, \ldots, Y_n)| \mid
    \NF(y;F) \in P\} \to |Y_k|\]
  as follows: let $\NF(y;F) = \nofo{X_1, \ldots, X_n}{x} \in P$, then
  $F(\iota_1, \ldots, \iota_n)(x) = y$ for some (unique) embeddings $\iota_i :
  X_i \hookrightarrow Y_i$; one then takes $f^F_{Y_1, \ldots, Y_n}(y)$ to be
  $\iota_k(f(x))$.

  By direct image, this induces a family of functions
  \[f^F_{Y_1, \ldots, Y_n} : \{S \subsetfin^* |F(Y_1, \ldots, Y_n)| \mid
    \NF(S;F) = P \} \to \powerfin^*(|Y_k|)\]
  Note that this could be extended to \enquote{$\NF(S;F) \subseteq P$} but most
  uses of this direct image will happen with $\NF(S;F) = P$.
  
  We also write $\Proj(P) = \Proj_1(P) \cup \ldots \cup \Proj_n(P)$.
\end{definition}

\begin{notation}
  Given a hypercoherence $X$ and $S \subsetfin^* |X|$, we define $l_X(S)$ to be
  $\odot$ if $S$ is a singleton, $\ominus$ if $S \in \hcns(X)$, $\oplus$ if $S
  \in \hcns(X^\perp)$.

  (Mnemonic: $l_{A \oplus B}(S \cup T) = \oplus$ for all non-empty $S \subseteq
  |A|$ and $T \subseteq |B|$.)
\end{notation}

Our eventual goal is to specify variable hypercoherences using projections as
follows: given $S \subsetfin^* |F(\vec{Y})|$, $l_{F(\vec{Y})}(S)$ should
be determined by the $l_{Y_k}(f^F_{\vec{Y}}(S))$ for $f \in \Proj_k(\NF(S;F))$,
$k \in \{1, \ldots, n\}$, following a sort of \enquote{truth table}. However, to
make this specification canonical, one should ensure that all \enquote{rows} in
the table are meaningful, i.e.\ serve to determine the (in)coherence of at least
one $S \subsetfin^* |F(\vec{Y})|$ for some $\vec{Y}$. This is the purpose of the
following.

\begin{definition}
  Let $F$ be a $n$-parameter variable hypercoherence, and $P \subseteq \NF(F)$.
  We write $\Val(P)$ for the set of \emph{valuations} on $P$, that is, of
  functions $\Proj(P) \to \{\odot, \ominus, \oplus\}$.

  A valuation $v$ on $P$ is \emph{possible} when there exist hypercoherences
  $Y_1, \ldots, Y_n$ and a subset $S \subseteq |F(Y_1, \ldots, Y_n)|$ such that
  $\NF(S;F) = P$ and $v(f) = l_{Y_k}(f^F_{Y_1,\ldots,Y_n}(S))$ for $f \in
  \Proj_k(P,F)$. The set of possible valuations on $P$ is denoted $\PVal(P)$.
\end{definition}

\begin{definition}
  A \emph{specification by projections} of $F$ is a dependent function
  \[ \sigma_F : (P \in \powerfin^*(\NF(F))) \to
    \PVal(P) \to \{\odot, \ominus, \oplus\} \]
  such that, for all $Y_1, \ldots, Y_n \in \HCohL$ and  $S \subsetfin^*
  |F(Y_1, \ldots, Y_n)|$,
  \[ l_{F(Y_1,\ldots,Y_n)}(S) = \sigma_F(\NF(S;F), (f \in \Proj_k(\NF(S;F))
    \mapsto l_{Y_k}(f^F_{Y_1,\ldots,Y_n}(S)))) \]

  A variable hypercoherence is \emph{strongly finite} if it is weakly finite and
  admits a specification by projections. (Note that if a specification by
  projections exists for $F$, it is unique.)
\end{definition}

To justify the terminology, observe that for a weakly finite $F$, there are
finitely many $P \subsetfin^* \NF(F)$ and the sets $\Proj_k(P,F)$ are finite.
Therefore, if $F$ admits a specification by projections, then this specification
is a finite object, and so $\Gamma(F)$ is finitely described. This notion
successfully excludes pathological examples such as our $F_f$ above:

\begin{proposition}
  If $F : \HCohI^n \to \HCohI$ is strongly finite, then $\vec{X} \mapsto
  F(\vec{X})$ is computable.
\end{proposition}
\begin{proof}
  All projections are computable, so it suffices to precompute a table encoding
  the specification, and to look up the relevant row.
\end{proof}

\begin{remark}
  Our definition of specification by projections is very restrictive. For
  instance, if $\nofo{\varnothing, \ldots, \varnothing}{x} \in P \subsetfin^*
  \NF(F)$ and $\Card(P) \geq 2$, then the $S \subsetfin^* |F(\vec{Y})|$ such
  that $\NF(S;F) = P$ are either all coherent or all incoherent, independently
  of $\vec{Y}$.
\end{remark}

We still need to show that strongly finite variable hypercoherences are closed
under all \malltwo connectives, and that the interpretations of \malltwo
formulas and proofs can be effectively computed. The main lemma is:

\begin{lemma}
  \label{lem:spec}
  There exists a criterion to determine, given $P \subsetfin^* \NF(F)$ and $v
  \in \Val(P)$, whether the valuation $v$ is possible; this criterion is
  effective when $F$ is weakly finite. Note that $F$ is not part of the input;
  that means that $\PVal(P)$ depends only on $P$, not on any other information
  on $F$.
\end{lemma}

\begin{proposition}
  \label{prop:constant-spec}
  The normal functors $(X_1, \ldots, X_k) \mapsto A$ and $(X_1, \ldots,
  X_n) \mapsto X_k$ admit specifications by projections (and are therefore
  strongly finite).
\end{proposition}

\begin{proposition}
  \label{prop:connectives-spec}
  If the $n$-parameter variable hypercoherences $F$ and $G$ can be specified by
  projections, then it is also the case for $F^\perp$, $F \otimes G$ and $F
  \oplus G$. Furthermore, if $F$ and $G$ are finite, then $\sigma_{F^\perp}$,
  $\sigma_{F \otimes G}$ and $\sigma_{F \oplus G}$ are computable from
  $\sigma_F$ and $\sigma_G$.
\end{proposition}
In the above proposition, finiteness may refer to either weak or strong
finiteness: since we assume specifiability by projections, those two notions
become equivalent by definition.

\begin{proposition}
  \label{prop:quantifier-spec}
  If $F \in \HCohL(n+1)$ admits a specification by projections, then so can
  $\forall(F) \in \HCohL(n)$. Furthermore, the function $(\NF(F),\sigma_F)
  \mapsto (\NF(\forall(F)),\sigma_{\forall(F)})$ defined on finite $F$ is
  computable.
\end{proposition}

From these propositions, we see that $A \mapsto \denot{A}$ is computable in the
second-order hypercoherence model. Since $(\pi : A) \mapsto \denot{\pi}
\sqsubset \denot{A}$ is computable for essentially the same reasons as in
coherence spaces, we may conclude:

\begin{theorem}
  Strongly finite variable hypercoherences form an effective model of
  \malltwo.
\end{theorem}

\section{Application: characterizing regular languages}
\label{sec:regular}

Next, we illustrate the usefulness of our finite semantics of \malltwo on the
following theorem.

\begin{definition}
  We consider the \enquote{stratified} Church encoding of strings: $\Str = \forall
  X.\Str[X]$, where $\Str[X] =\oc(X\multimap X) \multimap \oc(X\multimap
  X)\multimap \oc(X\multimap X)$.

  Given a proof $\pi : \oc\Str \multimap \oc^{k}\Bool$ (with\footnote{We also
    take $\mathtt{true}$ (resp.\ $\mathtt{false}$) to be the proof of $1 \oplus
    1$ proving the left (resp.\ right) occurrence of $1$.} $\Bool = 1 \oplus
  1$), the language decided by $\pi$ is defined as\footnote{$\pi(\overline{w})$
    denotes the proof obtained by applying a cut to $\pi$ and $\overline{w}$,
    and $\oc^k \mathtt{true}$ is the proof consisting of $k$ promotion rules
    with empty context applied to $\mathtt{true}$; cf.\ the sequent calculus
    recalled in \cref{sec:appendix-regular}.} $\Lang(\pi) = \{ w \in \{0,1\}^*
  \mid \pi(\overline{w}) \longrightarrow^* \oc^{k}\mathtt{true} \}$, where
  $\overline{w}$ is the Church encoding of $w$.
\end{definition}

\begin{theorem}
  \label{thm:regularELL}
  The type $\oc\Str \multimap \oc\oc\Bool$ in \emph{second-order Elementary
    Linear Logic (\elltwo)} captures the class of \emph{regular languages}. In
  other words, the languages that can be expressed as $\Lang(\pi)$ for some
  proof $\pi : \oc\Str \multimap \oc\oc\Bool$ in \elltwo are exactly the regular
  languages.
\end{theorem}

Roughly speaking, \elltwo is a subsystem of full second-order linear logic
where the rules governing the exponential modalities are restricted: promotion
and dereliction are removed, and replaced with \emph{functorial promotion}: from
$\vdash A_1, \ldots, A_n, B$, infer $\vdash \wn A_1, \ldots, \wn A_n, \oc B$.
This induces a sort of \enquote{stratification} on formulas and proofs, which is
the reason why the number of $\oc$ modalities in the output type of an \elltwo
function is significant.

The formal definition of \elltwo is given in \cref{sec:appendix-regular}. Some
parts of the proof of \cref{thm:regularELL} are also relegated to this section
of the appendix. To summarize:
\begin{itemize}
\item By encoding deterministic finite automata as proofs of $\oc\Str \multimap
  \oc\oc\Bool$, we show that every regular language can be decided by such a
  proof.
\item Using the aforementioned stratification property of \elltwo, we reduce the
  converse (only regular languages can be decided) to the lemma below.
  This reduction involves a \enquote{truncation at depth $k$} operation, similar
  to the one defined in~\cite{ealreg} for the elementary affine
  $\lambda$-calculus, which might be of independent interest.
\end{itemize}
The lemma whose proof features coherence spaces actually applies to full
second-order linear logic (with unrestricted exponentials). This is because the
\enquote{geometric} properties specific to \elltwo have already been exploited
in the previous step.

\begin{lemma}
  \label{lem:semeval}
  Let $\pi : \Str[A_1] \otimes \ldots \otimes \Str[A_n] \multimap \Bool$ be a
  proof in second-order linear logic where $A_1, \ldots, A_n$ are closed
  \malltwo formulas. Then the following language is regular:
  \[ \{ w \in \{0,1\}^* \mid \pi(\overline{w}[A_i] \otimes \ldots \otimes
    \overline{w}[A_n]) \longrightarrow^* \mathtt{true} \} \]
  where $\overline{w}[A_i] : \Str[A_i]$ is the instantiation of $\overline{w}$
  on $A_i$.
\end{lemma}
\begin{proof}
  Let $B = \Str[A_1] \otimes \ldots \otimes \Str[A_n]$. By \cref{cor:fincoh}, we
  know that $\denot{B}$ is a finite coherence space. Indeed, if a subformula of
  $B$ is the form $\oc C$, then it cannot be a subformula of some $A_i$ since
  the $A_i$ are \malltwo formulas, so $C = A_i \multimap A_i$ for some $i \in
  \{1, \ldots, n\}$. Since $A_i$ is closed, the premise of \cref{cor:fincoh}
  holds.

  Let $x \in \{0,1\}$. The operation \enquote{add a $x$ at the end of the
    string} is definable by a proof $\mathtt{snoc}^x_X$ of $\Str[X] \multimap
  \Str[X]$. From this, we can derive $\mathtt{snoc}^x_{A_1,\ldots,A_n} : B
  \multimap B$. This allows us to define a deterministic finite automaton
  (writing $\varepsilon$ for the empty string):
  \begin{itemize}
  \item whose states are the cliques of $\denot{B}$, with initial state $q_I =
    \denot{\overline\varepsilon[A_1] \otimes \ldots \otimes
      \overline\varepsilon[A_n]}$;
  \item whose transition function is $\delta(x,q) = \denot{\mathtt{snoc}^x_{A_1,
        \ldots, A_n}}(q)$ for $x \in \{0,1\}$;
  \item whose accepting states are $\{ q \sqsubset B \mid \denot{\hat\pi}(q) =
    \denot{\mathtt{true}} \}$.
  \end{itemize}
  Thanks to the compositionality of the coherence space model, when the DFA
  reads a word $w \in \{0,1\}$, it ends in the state $\denot{\overline{w}}$.
  This state is accepting if and only if $\denot{\hat\pi(\overline{w})} =
  \denot{\mathtt{true}}$; since the semantics is injective on $\Bool$, the DFA
  recognizes the language we want.
\end{proof}

\section{Conclusion}

Motivated by applications to implicit complexity, we sought a finite semantics
for \malltwo, and obtained it by proving the finiteness of the pre-existing
model of coherence spaces and normal functors. In retrospect, this is not so
surprising: one advantage of coherence spaces (e.g.\ over Scott domains), that
had already been stressed early in their history, is their tendency to give
small and legible interpretations to formulas.

Another early observation by Girard was that the existential introduction in
this model has a non-trivial computational contents, subsuming the cut rule --
this was mentioned as being \enquote{key to a semantic approach to
  computation}~\cite[p.~57]{girardLL}. By going from $A[T]$ to $\exists X.\, A$,
the information of the witness $T$ is compressed into some bounded data, and
this is why the semantics can be finite. Let us reformulate this from the
programming language point of view on existential types as abstract data types:
the cliques of $\denot{\exists X.\, A}$ keep just enough information about the
cliques of $\denot{A[T]}$ they originate from to determine their interaction
with the generic (universally typed) programs which might use them.

Relatedly, observe that the syntactic model of propositional \mall is finite,
and the existential witnesses are the only reason why this is not the case in
\malltwo. One could also try to directly implement the above intuitions starting
from the syntax; this will be the subject of an upcoming paper with P.\ Pistone,
T.\ Seiller and L.\ Tortora de Falco.

Aside from finiteness, the present work also investigated in detail the question
of effectivity. Almost no effort is needed in the case of coherence spaces, but
to obtain an effective hypercoherence model of \malltwo, we had to introduce the
idea of specifications by projections.

\bibliography{bibliography}

\appendix

\section{Omitted proofs of \cref{sec:coh}}
\label{sec:appendix-coh}

\begin{proof}[Proof of \cref{prop:nf-finite}]
  ($\Longrightarrow$) There are finitely many coherence spaces of cardinality
  $\leq \deg(F)$, and their images by $F$ are all finite since $F$ preserves
  finiteness.

  ($\Longleftarrow$) $\deg(F)$ is the supremum of a finite subset of $\naturalN$
  and is therefore finite.
\end{proof}

\begin{proof}[Proof of \cref{thm:polynomialgrowth}]
  Note that $\Card(|F(X)|) = O(\Card(|X|)^d)$ implies that $F$ sends finite
  spaces to finite spaes. Thus, in the remainder of this proof, we assume that
  $F$ preserves finiteness of cardinality (otherwise, an equivalence between two
  false propositions is true). With this assumption, it suffices to prove that
  \[\deg(F) = \inf\{d \in \mathbb{N} \mid \Card(|F(X)|) = O(\Card(|X|)^d)\}\]
  We decompose this equality into two inequalities.

  ($\leq$) Let $\nofo{X^0}{x} \in \NF(F)$ and $d = \Card(|X^0|)$.
  Let $[n] = 1 \& \ldots \& 1$ ($n$ times). For all $n$, there are $n^d$
  embeddings $X^0 \hookrightarrow X^0 \otimes [n]$ which are the identity on the
  first component. If for two such embeddings $\iota$ and $\iota'$, $F(\iota)(x)
  = F(\iota')(x)$, then by uniqueness of the normal form $\iota$ and $\iota'$
  are isomorphic in the slice category $\CohI/(X^0 \otimes [n])$; in the
  commuting triangle $\iota = \rho \circ \iota'$, $\rho$ can only be the
  identity, since our considered family of embeddings differ only on the
  component $[n]$, and so $\iota = \iota'$. Thus, $\iota \mapsto F(\iota)(x)$ is
  injective over the $n^d$ embeddings we consider, and therefore $\Card(|F(X^0
  \otimes [n])|) \geq n^d$ while $\Card(|X^0 \otimes [n]|) = dn$.

  ($\geq$) We assume $d = \deg(F) < \infty$ (when $\deg(F) = \infty$, the
  inequality is true for trivial reasons). This entails that $\NF(F)$ is finite,
  by \cref{prop:nf-finite}. Now let $X$ be a finite coherence space of
  cardinality $n$; each point $x \in |F(X)|$ has a normal form, so
  $\Card(|F(X)|)$ can be bounded by summing over $\nofo{X^0}{x} \in \NF(F)$ the
  number of possible embeddings of $X^0$ in $X$. This number is at most
  $O(n^{\Card(X^0)})$ and by definition $\Card(X^0) \leq d$. In the end, using
  $\Card(\NF(F)) = O(1)$, we get $\Card(|F(X)|) = O(n^d)$.
\end{proof}

\begin{proof}[Proof of \cref{prop:degree-mallzero}]
The case of negation is straightforward as $\trame{F(X_1,\dots,
X_n)}=\trame{F^{\bot}(X_1,\dots,X_n)}$. For the $\oplus$ case, 
one easily checks that $\nf{F\oplus G}\cong \nf{F}\uplus \nf{G}$.

Only the $\otimes$ case needs to be carefully checked. Consider
$\nofo{\vec{Z}}{(x,y)} \in \nf{F\otimes G}$, and the corresponding
normal forms $\nofo{\vec{X}}{x} \in \nf{F}$ and $\nofo{\vec{Y}}{y} \in
\nf{G}$. Then one can show, from minimality of $\vec{Z}$, that\footnote{For
  legibility purposes, we assume that $F$ preserves inclusions, and work with
  inclusions instead of embeddings; also, all operations are
applied componentwise on the $n$-tuples $\vec{X}$, $\vec{Y}$, $\vec{Z}$.}
$\vec{X}\cup\vec{Y}\subseteq\vec{Z}$. Moreover, since $x\in
F(\vec{Z})$ and $y\in G(\vec{Z})$, we have $\vec{X}\cup\vec{Y}
\supseteq\vec{Z}$. Thus $\vec{X}\cup\vec{Y}=\vec{Z}$ and $\deg{
(F\otimes G)} \leqslant \deg{F} + \deg{G}$. The converse inequality
is obtained by noticing that if $\nofo{\vec{X}}{x}\in\nf{F}$ and 
$\nofo{\vec{Y}}{y} \in \nf{G}$, then $\nofo{\vec{X}\oplus\vec{Y}}{(x,y)}
\in \nf{F\otimes G}$.
\end{proof}

\begin{proof}[Proof of \cref{prop:degree-forall}]
  Suppose that $\nofo{Y_1,\dots,Y_n}{\nofo{X_1}{x}}\in\nf{\forall(F)}$. It
  suffices to check that $\nofo{Y_1,\dots,Y_n,X_1}{x}\in\nf{F}$.

  % TODO: vraie preuve
\end{proof}

% \begin{proof}[Proof of \cref{cor:fincoh}]

% \end{proof}

% \begin{proof}[Proof of \cref{prop:cohnf-decidable}]
%   Indeed, the definition of coherence on $\nf{F}$ involves a quantification over
%   all spaces in which the normal form(s) involved may embed, but this
%   quantification can be bounded as remarked in~\cite[Remark~C.3]{girardF}. (As
%   mentioned in the introduction, this was one of Girard's motivations for
%   restricting the model of \emph{qualitative domains} to its binary case, i.e.\
%   coherence spaces.)
% \end{proof}

% \begin{proof}[Proof of \cref{thm:effectivity}]
%   TODO
% \end{proof}

\begin{proof}[Proof of \cref{thm:logspace}]
  % TODO refer to instantiation

  For each point $x \in \theta_X$, there is a unique $(\iota,y)$ such that
  $F(\iota)(y) = x$; indeed, it is the initial object of the slice category
  $\El(\trame{F})/(X,x)$. So we have a bijection
   \[ \theta_X \cong \{ (\iota,y) \mid \nofo{Y}{y} \in \theta, \iota : Y \to X
     \text{ embedding} \} \]

   This bijection is computable in logarithmic space: it is just a matter of
   performing substitutions on terms of size $O(1)$ (since $F$ is fixed),
   although the $\iota$'s are not of constant size (because of the representation
   of elements of $\trame{X}$). Therefore the problem is reduced to enumerating
   the right-hand side without repetitions.

   There are finitely many $\nofo{Y}{y}$ in $\theta$. For each of them, each
   injection $Y \to X$ can be represented by $\card{\trame{Y}} \leq \deg(F) =
   O(1)$ elements of $\trame{X}$: this takes $O(\log \card{\trame{X}})$ space.
   Thus, all the injections can be enumerated in logarithmic space, and for each
   injection, whether it is an embedding can be determined in logarithmic space
   (using the coherence relation on $\trace{F}$ which may be precomputed
   independently of the input.)
\end{proof}

\section{Details for \cref{sec:affine}}
\label{sec:appendix-affine}

\section{Omitted proofs of \cref{sec:hypercoh}}
\label{sec:appendix-hypercoh}

\section{Details and proofs for \cref{sec:regular}}
\label{sec:appendix-regular}

% TODO use the Danos--Joinet--Laurent version of ELL (it actually works)

\begin{figure}[h]
  \[
    \text{(functorial promotion)}
    \frac{\vdash \Gamma, A}{\vdash \wn \Gamma, \oc A}
    \qquad
    \text{(weakening)}
    \frac{\vdash \Gamma}{\vdash \Gamma, \wn A}
    \qquad
    \text{(contraction)}
    \frac{\vdash \Gamma, \wn A, \wn A}{\vdash \Gamma, \wn A}
  \]
  \caption{Exponential rules for the $\elltwo$ sequent calculus. In the
    functorial promotion rule, when $\Gamma = B_1, \ldots, B_k$, $\wn \Gamma$
    stands for $\wn B_1, \ldots, \wn B_k$.}
  \label{fig:ell-rules}
\end{figure}

\begin{figure}[h]
  \centering
 \begin{gather*} 
     \text{(\texttt{ax}-rule)}
     \frac{}{\vdash A, A^\bot}
    \qquad
     \text{(cut rule)}
     \frac{\vdash \Gamma, A \quad \vdash A^\bot, \Delta}{\vdash \Gamma, \Delta}
    \qquad
     \text{(exchange rule)}
     \frac{\vdash \Gamma, A, B, \Delta}{\vdash \Gamma, B, A, \Delta}
    \\ \\
    \text{($\otimes$-rule)}
    \frac{\vdash \Gamma, A \quad \vdash B, \Delta}{\vdash \Gamma, A \otimes B,
      \Delta}
    \qquad
    \text{($\parr$-rule)}
    \frac{\vdash \Gamma, A, B}{\vdash \Gamma, A \parr B}
    \qquad
    \text{($\bot$-rule)}
    \frac{\vdash \Gamma}{\vdash \Gamma, \bot}
    \qquad
    \text{($1$-rule)}
    \frac{}{\vdash 1}
   \\ \\
    \text{($\oplus$-rule)}
    \frac{\vdash \Gamma, A_i}{\vdash \Gamma, A_1 \oplus A_2}
    \text{\footnotesize for $i \in \{1,2\}$}
    \qquad
    \text{($\with$-rule)}
    \frac{\vdash \Gamma, A \quad \vdash \Gamma, B}{\vdash \Gamma, A \with B
      }
    \qquad
    \text{($\top$-rule)}
    \frac{}{\vdash \Gamma, \top}
   \\ \\
    \text{($\exists$-rule)}
    \frac{\vdash \Gamma, A[B/X]}{\vdash \Gamma, \exists X.\, A}
    \qquad
    \text{($\forall$-rule)}
    \frac{\vdash \Gamma, A}{\vdash \Gamma, \forall X.\, A}
    \text{\footnotesize for $X$ not free in $\Gamma$}
 \end{gather*}
\caption{Rules for the $\malltwo$ sequent calculus (there is no rule for $0$).}
\label{fig:mall2-rules}
\end{figure}

\begin{figure}
\[
\def\arraystretch{4}
\begin{array}{r !\qquad !\leadsto !\qquad l}
\text{
\AXC{ }
\UIC{$A, A^\perp$}
\AXC{$\pi$}
\UIC{$A,\Delta$}
\BIC{$A,\Delta$}
\DisplayProof}
&
\text{
\AXC{$\pi$}
\UIC{$A,\Delta$}
\DisplayProof}
\\
\text{
\AXC{$\pi_1$}
\UIC{$\vdash \Gamma, A$}
\AXC{$\pi_2$}
\UIC{$\vdash \Gamma', B$}
\BIC{$\vdash \Gamma, \Gamma', A \otimes B$}
\AXC{$\pi_3$}
\UIC{$\vdash A^\perp, B^\perp, \Delta$}
\UIC{$\vdash A^\perp \parr B^\perp, \Delta$}
\BIC{$\vdash \Gamma, \Gamma', \Delta $}
\DisplayProof}
&
\text{
\AXC{$\pi_1$}
\UIC{$\vdash \Gamma, A$}
\AXC{$\pi_2$}
\UIC{$\vdash \Gamma', B$}
\AXC{$\pi_3$}
\UIC{$\vdash B^\perp, A^\perp, \Delta$}
\BIC{$\vdash A^\perp, \Gamma', \Delta$}
\BIC{$\vdash \Gamma, \Gamma', \Delta $}
\DisplayProof}
\\
\text{
\AXC{$\pi_1$}
\UIC{$\vdash \Gamma, A$}
\UIC{$\vdash \Gamma, A \oplus B$}
\AXC{$\pi_2$}
\UIC{$\vdash A^\perp, \Delta$}
\AXC{$\pi_3$}
\UIC{$\vdash B^\perp, \Delta$}
\BIC{$\vdash A^\perp \with B^\perp, \Delta$}
\BIC{$\vdash \Gamma, \Delta $}
\DisplayProof}
&
\text{
\AXC{$\pi_1$}
\UIC{$\vdash \Gamma, A$}
\AXC{$\pi_2$}
\UIC{$\vdash A^\perp, \Delta$}
\BIC{$\vdash \Gamma, \Delta $}
\DisplayProof}
\\
\text{
\AXC{$\pi_1$}
\UIC{$\vdash \Gamma, A$}
\UIC{$\vdash \wn \Gamma, \oc A$}
\AXC{$\pi_2$}
\UIC{$\vdash A^\perp, \Delta, B$}
\UIC{$\vdash \wn A^\perp, \wn \Delta, \oc B$}
\BIC{$\vdash \wn \Gamma, \wn \Delta, \oc B $}
\DisplayProof}
&
\text{
\AXC{$\pi_1$}
\UIC{$\vdash \Gamma, A$}
\AXC{$\pi_2$}
\UIC{$\vdash A^\perp, \Delta, B$}
\BIC{$\vdash \Gamma, \Delta, B $}
\UIC{$\vdash \wn \Gamma, \wn \Delta, \oc B $}
\DisplayProof}
\\
\text{
\AXC{$\pi_1$}
\UIC{$\vdash \wn \Gamma, \oc A$}
\AXC{$\pi_2$}
\UIC{$\vdash \Delta$}
\UIC{$\vdash \wn A^\perp, \Delta$}
\BIC{$\vdash \wn \Gamma, \Delta $}
\DisplayProof}
&
\text{
\AXC{$\pi_2$}
\UIC{$\vdash \Delta $}
\UIC{$\vdash \wn \Gamma, \Delta $}
\DisplayProof}
\\
\text{
\AXC{$\pi_1$}
\UIC{$\vdash \wn \Gamma, \oc A$}
\AXC{$\pi_2$}
\UIC{$\vdash \wn A^\perp, \wn A^\perp, \Delta$}
\UIC{$\vdash \wn A^\perp, \Delta$}
\BIC{$\vdash \wn \Gamma, \Delta $}
\DisplayProof}
&
\text{
\AXC{$\pi_1$}
\UIC{$\vdash \wn \Gamma, \oc A$}
\AXC{$\pi_1$}
\UIC{$\vdash \wn \Gamma, \oc A$}
\AXC{$\pi_2$}
\UIC{$\vdash \wn A^\perp, \wn A^\perp, \Delta$}
\BIC{$\vdash \wn A^\perp, \wn \Gamma, \Delta$}
\BIC{$\vdash \wn \Gamma, \wn \Gamma, \Delta $}
\UIC{$\vdash \wn \Gamma, \Delta $}
\DisplayProof}
\\
\text{
\AXC{$\pi_1$}
\UIC{$\vdash \Gamma, A$}
\UIC{$\vdash \Gamma, \forall X.\, A$}
\AXC{$\pi_2$}
\UIC{$\vdash A^\perp[B/X], \Delta$}
\UIC{$\vdash \exists X.\, A^\perp, \Delta$}
\BIC{$\vdash \Gamma, \Delta $}
\DisplayProof}
&
\text{
\AXC{$\pi_1[B/X]$}
\UIC{$\vdash \Gamma, A[B/X]$}
\AXC{$\pi_2$}
\UIC{$\vdash A^\perp[B/X], \Delta$}
\BIC{$\vdash \Gamma, \Delta $}
\DisplayProof}
\\
\end{array}
\]

\caption{Key reductions of $\elltwo$ cut-elimination.}
\label{fig:ell2-keycutelim}
\end{figure}

\subsection{Proof of extensional completeness}

\begin{proposition}
  Any regular language can be expressed as $\Lang(\pi)$ for some \elltwo proof
  $\pi : \oc\Str \multimap \oc\oc\Bool$.
\end{proposition}
\begin{proof}[Proof sketch]
  To encode a deterministic finite automaton with set of states $Q$ and
  transition function, simply instantiate the input string at the type $1 \oplus
  \ldots \oplus 1$ ($\Card(Q)$ times), and give it as arguments the linear
  functions representing $\delta(0,-)$ and $\delta(1,-)$.
\end{proof}

\subsection{Proof of soundness by reduction to \cref{lem:semeval}}

In this section, we fix $\pi : \oc\Str \multimap \oc\oc\Bool$ in \elltwo, and
prove the converse of the above. The first step is to understand the shape of
$\pi$.

\begin{lemma}\label{lem:predicates}
Up to commutations, $\pi$ is of the form
\vspace{-0.2cm}
\begin{prooftree}
\AxiomC{{$\hat{\pi}$}}
\UnaryInfC{$\vdash\Str[A_1]^\bot,\dots,\Str[A_n]^\bot,\oc\Bool$}
\doubleLine
\UnaryInfC{$\vdash\Str^\bot,\dots,\Str^\bot,\oc\Bool$}
\RightLabel{\scriptsize{$\oc$}}
\UnaryInfC{$\vdash\wn\Str^\bot,\dots,\wn\Str^\bot,\oc\oc\Bool$}
\doubleLine
\UnaryInfC{$\vdash\wn\Str^\bot,\oc\oc\Bool$}
\RightLabel{\scriptsize{$\parr$}}
\UnaryInfC{$\vdash\oc\Str\multimap\oc\oc\Bool$}
\end{prooftree}
\end{lemma}

\begin{proof}
   A proof of $\vdash\wn\Str^\bot,\dots,\wn\Str^\bot,\oc\oc\Bool$ 
 necessarily ends either with a structural rule or a promotion. 
 From this and the invertibility of $\parr$, one obtains that the
 proof, up to commutation, ends with the following sequence
 of rules:
 \begin{prooftree}
 \AxiomC{$\vdash\Str^\bot,\dots,\Str^\bot,\oc\Bool$}
 \RightLabel{\scriptsize{$\oc$}}
 \UnaryInfC{$\vdash\wn\Str^\bot,\dots,\wn\Str^\bot,\oc\oc\Bool$}
 \doubleLine
 \UnaryInfC{$\vdash\wn\Str^\bot,\oc\oc\Bool$}
 \RightLabel{\scriptsize{$\parr$}}
 \UnaryInfC{$\vdash\oc\Str\multimap\oc\oc\Bool$}
 \end{prooftree}
 Now, recall that $\Str^\bot=\exists X. \Str[X]^\bot$. Moreover, 
 notice the introduction rule for $\exists$ commutes with all other
 rules except promotion and the introduction of $\forall$. We only
 need to show that the introduction rule for all $\exists$ connectives
 of the occurrences of $\Str^\bot$ are not followed by a promotion 
 or a $\forall$ introduction. Ruling out promotion is easy, as all 
 formulas in the conclusion sequent of a promotion rule have 
 exponential as principal connectives. Moreover, it is possible to 
 rule out $\forall$ introductions as folllws. If a $\forall$ introduction
 rule appears in the proof, the $\forall$ connective it introduces 
 is part of an existential witness. But all variables existentially
 quantified appear under the scope of an exponential connective,
 and therefore can only precede a promotion rule.
\end{proof}

Now, a crucial observation is that \emph{the $A_i$ in the previous lemma can be
  taken in \malltwo w.l.o.g.} This is where the \emph{stratification} property
of \elltwo plays a key role.

\begin{lemma}\label{lem:predicatesMALL}
  There is a proof $\pi'$ whose witnesses $A_1, \ldots, A_n$ are closed \malltwo
  formulas, and which decides the same language as $\pi$.
\end{lemma}
\begin{proof}
   We define the \emph{truncation at depth 2} of a formula as follows: all subformulas
   of the form $!A$ (resp.\ $?A$) at depth 2, i.e. in the scope of two other
   nested $!/?$ modalities, are replaced by $1$ (resp.\ $\bot$). Note that the
   truncation at depth 2 of $!\Str$ and $!!\Bool$ are themselves.

   This operation extends to proofs: any functorial promotion of conclusion
   $\vdash \wn B_1, \ldots, \wn B_m, \oc C$ is replaced by the only proof of
   $\vdash \bot, \ldots, \bot, 1$, while contractions and weakenings are replaced
   by cuts with $1 \vdash 1 \otimes 1$ and $\vdash 1$. Note that this truncation is
   the identity on cut-free proofs of $!!\Bool$ and $!\Str$.

   One may then check that truncation at depth 2 is compatible with cut-elimination,
   which means that one can replace $\pi$ by its truncation at depth 2 and still
   recognize the same language. Then the $A_i$ are replaced by their ``truncation
   at depth 0'' which are \malltwo formulas.
   % TODO: closed
\end{proof}

\begin{proposition}
  $\Lang(\pi)$ is regular.
\end{proposition}
\begin{proof}
  By \cref{lem:semeval}, the language
  \[ \{ w \in \{0,1\}^* \mid \hat\pi(\overline{w}[A_i] \otimes \ldots \otimes
    \overline{w}[A_n]) \longrightarrow^* \mathtt{true} \} \]
  is regular. A examination of the cut-elimination process reveals
  that this language is none other than $\Lang(\pi)$.
\end{proof}
  
\end{document}